\documentclass{article}
\usepackage{amsmath,amssymb,graphicx,amsfonts}
\usepackage{algorithmic}
\usepackage{algorithm}
\usepackage{tikz,xcolor}
\usepackage{graphicx} 
\usetikzlibrary{backgrounds,decorations.pathmorphing,arrows}
\usepackage{amsopn}

\def\EE{\mathbb E}

\def\RR{\mathbb R}
\def\QQ{\mathbb Q}
\def\ZZ{\mathbb Z}

\def\cB{\mathcal B}

\def\cG{\mathcal G}

\def\cS{\mathcal S}

\def\cT{\mathcal T}

\def\bx{\mathbf x}

\def\bx{\mathbf x}

\def\bp{\mathbf p}

\def\b1{\mathbf 1}

\def\eps{\varepsilon}

\newcommand{\qed}{\hfill$\square$\bigskip}
\newcommand{\raf}[1]{(\ref{#1})}
\newcommand{\proof}{\noindent {\bf Proof}~~}

\newcommand{\poly}{\operatorname{poly}}

\newcommand{\Max}{\textsc{White}}
\newcommand{\Min}{\textsc{Black}}

\newcommand{\vol}{\operatorname{vol}}

\newcommand{\hide}[1]{}

\def\Max{\textsc{Max}}
\def\Min{\textsc{Min}}

\newcommand{\bone}{\ensuremath{\boldsymbol{1}}}

\newtheorem{theorem}{Theorem}

\newtheorem{lemma}{Lemma}

\newtheorem{claim}{Claim}
\newtheorem{proposition}{Proposition}

\newtheorem{fact}{Fact}

\newtheorem{remark}{Remark}

\title{A Convex Programming-based Algorithm for Mean Payoff Stochastic Games with Perfect Information\thanks{Part of this research was done at the Mathematisches Forschungsinstitut Oberwolfach during a stay within the "Research in Pairs" Program from July 26, 2015-August 15, 2015.
	The research of the third author has been partially funded by the Russian Academic Excellence Project '5-100'. The fourth author was partially supported by KAKENHI Grant Numbers
	24106002 and 26280001.
	}}
\author{
	Endre Boros\thanks{MSIS Department and RUTCOR, Rutgers University, 100 Rockafellar Road, Livingston Campus
		Piscataway, NJ 08854, USA;
		(Endre.Boros@rutgers.edu)}
	\and
	Khaled Elbassioni\thanks{Masdar Institute of Science and Technology, P.O.Box 54224, Abu Dhabi, UAE;
		(kelbassioni@masdar.ac.ae)}
	\and
	Vladimir Gurvich\thanks{MSIS Department and RUTCOR, Rutgers University, 100 Rockafellar Road, Livingston Campus
		Piscataway, NJ 08854, USA; and National Research University, Higher School of Economics, Moscow, Russia;
		(Vladimir.Gurvich@gmail.com)
	}
	\and
	Kazuhisa Makino\thanks{Research Institute for Mathematical Sciences (RIMS)
		Kyoto University, Kyoto 606-8502, Japan;
		(makino@kurims.kyoto-u.ac.jp)}
}

\begin{document}
\date{}
\maketitle
\begin{abstract}
We consider two-person zero-sum stochastic mean payoff games
with perfect information, or BWR-games, given by a digraph
$G = (V, E)$, with local rewards $r: E \to \ZZ$, and
three types of positions: black $V_B$, white $V_W$, and random $V_R$ forming a partition of $V$.
It is a long-standing open question whether a polynomial time algorithm for BWR-games exists, even when $|V_R|=0$. 
In fact, a pseudo-polynomial algorithm for BWR-games would already
imply their polynomial solvability.
In this short note, we show that BWR-games can be solved via convex programming in pseudo-polynomial time if the number of random positions is  a constant. 
\end{abstract}

\section{Introduction}
\label{intro}
We consider two-person zero-sum
stochastic games with perfect information and mean payoff:
Let  $G = (V, E)$ be a digraph whose vertex-set  $V$
is partitioned into three subsets  $V = V_B \cup V_W \cup V_R$
that correspond to black, white, and random positions, controlled
respectively, by two players, \Min\ - the \emph{minimizer}  and
\Max\ - the \emph{maximizer}, and by nature.
We also fix a {\em local reward} function  $r: E \to \ZZ$, and probabilities
$p(v,u)>0$ for all arcs  $(v,u)$ going out of  $v \in V_R$. We assume that  $\sum_{u\mid (v,u)\in E}p(v,u)=1$, for all $v\in V_R$.
Vertices  $v \in V$  and  arcs $e \in E$  are called
{\em positions} and {\em moves}, respectively.
The game begins at time $t=0$ in the initial position $s_0=v_0$. In a general step, in time $t$, we are at position $s_t\in V$.
The player who controls $s_t$ chooses an outgoing arc $e_{t+1}=(s_t,v)\in E$, and the game moves to position $s_{t+1}=v$. If $s_t\in V_R$ then an outgoing arc is choses with the given probability $p(s_t,s_{t+1})$.  
We assume that every vertex in $G$ has an outgoing arc.
In general, the strategy of the player is a policy by which (s)he chooses the outgoing arcs from the vertices (s)he controls. This policy may involve the knowledge of the previous steps as well as probabilistic decisions. We call a strategy {\it stationary} if it does not depend on the history and {\it pure} if it does not involve probabilistic decisions. For this type of games, it will be enough to consider only such strategies, since these games are known to be (polynomially) equivalent \cite{DGAA13} to the perfect information stochastic games considered by Gillette \cite{Gil57,LL69}. 

In the course of this game players and nature generate an infinite sequence
of edges $\bp=(e_1,e_2,\ldots)$ (a \emph{play}) and the corresponding real sequence
$r(\bp)=(r(e_1),r(e_2),\ldots)$ of local rewards.
There is a global payoff function $\phi$ that maps any local reward sequence to a real number, and it is assumed that \Min\ pays \Max\ the  amount $\phi(r(\bp))$ resulting from the play.
Naturally, \Max's aim is to create a play which maximizes $\phi(r(\bp))$, while \Min\ tries to minimize it.
(Let us note that
the local reward function $r:E\rightarrow \RR$ may have negative values, and
$\phi(r(\bp))$ may also be negative, in which case
\Max\ has to pay \Min\ $-\phi(r(\bp))$. Let us also note that $r(\bp)$ is a random variable since random transitions occur at positions in $V_R$.) 
Here $\phi$ stands for the \emph{limiting mean payoff}
\begin{equation}\label{e0}
\phi(r(\bp))=\liminf_{T\rightarrow\infty}\frac{\sum_{i=1}^T\EE[r(e_i)]}{T},
\end{equation}
where $\EE[r(e_i)]$ is the expected reward incurred at step $i$ of the play.

As usual, a pair of (not necessarily pure or stationary) strategies is a {\it saddle point} (or {\it equilibrium}) if neither of the players can improve individually by changing her/his strategy. The corresponding $\phi(r(\bp))$ is the value $\mu^{}_{\cG}(v_0)$ of the game with respect to initial position $v_0$. Such a pair of strategies are called {\it optimal}; furthermore, it is called {\it uniformly optimal} if it provides the value of the game for any initial position. It is known \cite{Gil57,LL69} that
every such game has a pair of uniformly optimal pure stationary strategies. A BWR-game is said to be {\it ergodic} if  $\mu^{}_\cG(v)=\mu$ for all $v\in V$, that is, the value is the same from each initial position.

\medskip

This class of \emph{BWR-games} was introduced in \cite{GKK88}; see also \cite{CH08}.
The special case when $V_R = \emptyset$, \emph{BW-games}, is also known as
{\em cyclic} games.
They were introduced for the complete bipartite digraphs in \cite{Mou76,Mou76a},
for all (not necessarily complete) bipartite digraphs in \cite{EM79}, and for arbitrary digraphs\footnote{In fact, BW-games on arbitrary digraphs can be polynomially reduced to BW-games on bipartite digraphs \cite{DGAA13}; moreover, the latter class can further be reduced to BW-games on complete bipartite digraphs \cite{Vienna}.} in \cite{GKK88}.
A more special case was considered extensively in the literature under
the name of \emph{parity games} \cite{BV01,BV01a,CJH04,Hal07,J98,JPZ06}, and later generalized
also to include random positions in \cite{CH08}.
A BWR-game is reduced to a \emph{minimum mean cycle problem} in case $V_W = V_R = \emptyset$,
see, e.g., \cite{Kar78}.
If one of the sets $V_B$ or $V_W$ is empty, we obtain
a {\em Markov decision process} (MDP), which can be expressed as a linear program; see, e.g., \cite{MO70}.
Finally, if both are empty, $V_B = V_W = \emptyset$, we get a {\em weighted Markov chain}.
For BW-games several pseudo-polynomial and subexponential algorithms are known 
\cite{GKK88,KL93,ZP96,Pis99,BV01,BV01a,BSV04,BV05,BV07,Hal07,Schewe09,Vor08};
see also \cite{JPZ06} for parity games.
Besides their many applications (see e.g. \cite{Lit96,J2000}), all these games are of interest to Complexity Theory: It is known \cite{KL93,ZP96} that
the decision problem ``whether the value of a BW-game is positive" is in the intersection of NP and co-NP. Yet, no polynomial algorithm is known for these games, see e.g., the survey by Vorobyov \cite{Vor08}. A similar complexity claim can be shown to hold for BWR-games, see \cite{AM09,DGAA13}.

\subsection*{Main result}

The computational complexity of stochastic games with perfect information is an  outstanding open question; see, e.g., the survey \cite{RF91}. 
While there are numerous pseudo-polynomial algorithms known for the BW-case, it is a challenging open question whether a pseudo-polynomial algorithm exists for BWR-games, as the existence of such an algorithm would imply also the polynomial solvability of this class of games \cite{AM09}.

In \cite{BEGM13,BEGM15-arxiv}, we gave a pseudo-polynomial algorithm for BWR-games when {\it the number of random positions is fixed}. In this note we show that one can obtain a similar result via convex programming, combined with some of the ideas in \cite{BEGM13,BEGM15-arxiv}.

For a BWR-game $\cG$ let us denote by $n=|V_W|+|V_B|+|V_R|$ the number of positions, by $k=|V_R|$ the number of random positions, and assume (without loss of generality) that all local rewards are non-negative integers with maximum value $U$ and all transition probabilities are rational with common denominator $D$. The main result of this paper is as follows.

\begin{theorem} \label{t-main}
A BWR-game $\cG$ can be solved in $\poly(n,U,D^k)$ time via convex programming.
\end{theorem}

This theorem extends the result by Schewe \cite{Schewe09}, where it was shown that solving BW-games can be reduced to solving linear programming problems with \emph{pseudo-polynomial} bit length.

\medskip

According to the results in \cite{BEGM13,BEGM15-arxiv}, to get a pseudo-polynomial algorithm for BWR-games, it is enough to have pseudo-polynomial routines for: (i) solving BW-games; (ii) solving {\it ergodic} BWR-games; and (iii) finding the top and bottom classes in a non-ergodic BWR-game (that is, the sets of positions with highest and lowest values). 

There are several pseudo-polynomial algorithms for solving BW-games,  e.g., \cite{GKK88,Pis99,ZP96}. One may also use the LP-based algorithm given in \cite{Schewe09}.  For (ii) we show in Section~\ref{sec111} how to obtain the top (resp., bottom) class in a BWR-game, and a pair of strategies solving the game induced by the top (resp., bottom) class. This also provides an algorithm for ergodic BWR-games as required in (iii).

\section{Potential transformations and canonical forms}
\label{ssPot}
Given a BWR-game  $\cG = (G, p, r)$, let us introduce
a mapping  $x : V \rightarrow \RR$, whose
values  $x(v)$  will be called {\em potentials}, and
define the transformed reward function $r_x:E\to\RR$  as:
\begin{equation}
\label{eqpot}
r_x(v,u) = r(v,u) + x(v) - x(u),
\;\; \mbox{where} \;\;
(v,u) \in E.
\end{equation}

\smallskip

It is not difficult to verify that the obtained game  $\cG^x$
and the original game  $\cG$ are equivalent (see \cite{DGAA13}).
In particular, their optimal
(pure stationary) strategies coincide, and
their value functions also coincide: $\mu^{}_{\cG^x} = \mu^{}_{\cG}$.

It is known that for BW-games there exists a potential transformation such that, in the obtained game the locally optimal strategies are globally optimal,
and hence, the value and optimal strategies become obvious \cite{GKK88}. This result was extended for the more general class of BWR-games in \cite{DGAA13}: in the transformed game, the equilibrium value $\mu^{}_{\cG}(v)=\mu^{}_{\cG^x}(v)$ is given simply by the maximum local reward for $v\in V_W$, the minimum local reward for $v\in V_B$, and the average local reward for $v\in V_R$. In this case we say that the transformed game is in \emph{canonical} form. To define this more formally, let us use the following notation throughout this section: Given functions $f:E\to\RR$ and $g:V\to\RR$, we define the functions ${M}[f],{M}[g]:V\to\RR$.

\begin{equation*}
{M}[f](v) = \left\{\begin{array}{lll}
\max_{u \mid (v, u) \in E}f(v,u), &\text{for } v \in V_W, \\
\min_{u \mid (v, u) \in E}f(v,u), &\text{for } v \in V_B, \\
\sum_{u \mid (v, u) \in E}   p(v,u)\, f(v,u), &\text{for } v \in V_R.
\end{array}\right.
\end{equation*}

\begin{equation*}
{M}[g](v) = \left\{\begin{array}{lll}
\max_{u \mid (v, u) \in E}g(u), &\text{for } v \in V_W, \\
\min_{u \mid (v, u) \in E}g(u), &\text{for } v \in V_B, \\
\sum_{u \mid (v, u) \in E}   p(v,u)\, g(u), &\text{for } v \in V_R.
\end{array}\right.
\end{equation*}

We say that a BWR-game $\cG$ is in (strong) canonical form if there exist vectors $\mu,x\in\RR^V$ such that
\begin{itemize}
\item[(C1)]   $\mu=  M[\mu]=M[r_x]$ and,
\item[(C2)] for every  $v \in V_W \cup V_B$,
every move  $(v,u) \in E$
such that  $\mu(v) =  r_x(v,u)$ must also have $\mu(v)=\mu(u)$, or
in other words, every locally optimal move $(v,u)$
is globally optimal.
\end{itemize}



\begin{theorem}[\cite{DGAA13}]
\label{t2}
For each BWR-game  $\cG$ there
is a potential transformation 
$x\in\RR^V$ that brings $\cG$  to
canonical form with $\|x\|_{\infty}\leq L:=n Uk(2D)^{k}$. Furthermore, in a game in canonical form we have $\mu^{}_{\cG} = M[r_x]$.
\end{theorem}

%
%
%

In this paper, we will provide a convex programming formulation based on the existence of potential transformations.

We will need the following upper bound on the required accuracy.
\begin{lemma}[\cite{IPCO-2010,BEGM15-arxiv}]\label{l1-prob}
For any position $v$ in the top (resp., bottom) class in a BWR-game $\cG$, the value $\mu^{}_{\cG}(v)$ is a rational number with a denominator at most $\sqrt{k}2^{k/2}D^{k+1}$. 
\end{lemma}


\begin{lemma}\label{l1}
Consider a BWR-game $\cG$ and denote by $\bone$ the vector of all ones. Then there exists a potential vector $x\in\RR^V$ and $t\in\RR$ such that $M[r_x] \ge t\bone$ if and only if $\mu^{}_{\cG}\ge t\bone$.	
\end{lemma}
\proof
	Indeed, if \Max\ (\Min) applies a locally optimal strategy $s_W^{}$ in the transformed game $\cG^x$ then after every own move (s)he will get (pay) at least $t$,  while for each move of the opponent the local reward will be at least (at most)  $t$, and finally, for each random position the expected local reward is  at least $t$.
	Thus, the expected local reward $\EE[r_x(e_i)]$ at each step of the play is at least $t$. Hence, by \raf{e0}, strategy $s_W^{}$ guarantees $\Max$ at least $t$ from any starting position.  
	
	The other direction follows from Theorem~\ref{t2}.
\qed

A symmetric version of Lemma~\ref{l1} can also be obtained by similar arguments.
\begin{lemma}\label{l2}
	Consider a BWR-game $\cG$. Then there exists a potential vector $x\in\RR^V$ and $t\in\RR$ such that $M[r_x] \le t\bone$ if and only if $\mu^{}_{\cG}\le t\bone$.	
\end{lemma}

\section{The convex programs}
The following simple facts relate the {\it softmax} (resp., {\it softmin}) to the maximimum (resp., minimum) of a set of numbers. 
\begin{fact}\label{softmax-softmin}
	For any numbers $a_1,\ldots,a_n\in\RR$ and $b>1$:
	\begin{itemize}
		\item[(i)] $\max_i{a_i}\le\log_b\sum_{i}b^{a_i}\le \max_ia_i+\log_b n$;
		\item[(ii)] $\min_i{a_i}\ge -\log_b\sum_{i}b^{-a_i}\ge \min_ia_i-\log_b n$.
	\end{itemize} 
\end{fact}
\proof
This follows from the fact the trivial inequalities $b^{\max_i{a_i}}\le\sum_{i}b^{a_i}\le n b^{\max_i{a_i}}$.  
\qed

\begin{fact}\label{concave}
Let $\alpha_1,\ldots,\alpha_n>0$ be given numbers such that $\sum_{i=1}^n\alpha_i=1$. Then the function $f(x)=\prod_{i=1}^nx_i^{\alpha_i}$ is concave for $x\ge 0$.
\end{fact}
\proof
Note that for any $x,y\in\RR^n_+$, if for some $i$, $x_i=0$ then for any $\lambda\in[0,1]$, 
\begin{align*}
\lambda f(x)+(1-\lambda)f(y)=&(1-\lambda)f(y)\\
=&(1-\lambda)\prod_{i=1}^ny_i^{\alpha_i}=\prod_{i=1}^n\left((1-\lambda)y_i\right)^{\alpha_i}\\
\le &\prod_{i=1}^n\left(\lambda_ix_i+(1-\lambda)y_i\right)^{\alpha_i}= f(\lambda\bx+(1-\lambda)y).
\end{align*}
Thus, it is enough to show that $\nabla^2f(x)$ is a negative semi-definite matrix for $\bx>0$. Note that
\begin{align*}
\frac{\partial f}{\partial x_i}=&\frac{\alpha_i}{x_i}f(x),~~\text{for }i=1,\ldots,n\\
\frac{\partial^2 f}{\partial x_i^2}=&\frac{\alpha_i(\alpha_i-1)}{x_i^2}f(x),~~\text{for }i=1,\ldots,n\\
\frac{\partial^2 f}{\partial x_i\partial x_j}=&\frac{\alpha_i\alpha_j}{x_ix_j}f(x),~~\text{for }i,j=1,\ldots,n,~i\ne j.
\end{align*}
Consider any $y\in\RR^n$. Then
\begin{align*}
y^T\nabla^2f(x)y&=\left(\sum_i\alpha_i(\alpha_i-1)\frac{y_i^2}{x_i^2}+\sum_{i\ne j}\alpha_i\alpha_j\frac{y_iy_j}{x_ix_j}\right)f(x)\\
&=\left(\left(\sum_i\alpha_i\frac{y_i}{x_i}\right)^2-\sum_i\alpha_i\left(\frac{y_i}{x_i}\right)^2\right)f(x)\le 0,
\end{align*}
 where the last inequality follows from Jensen's inequality applied to the convex function $f(w)=w^2$. 
 
\qed
 
Given $t\in \RR$, let us replace the max operator in the system $M[r_x]\ge t$ by the softmax approximation:
\begin{align}\label{e1-1}
\log_b\sum_{u \mid (v, u) \in E}b^{r(v,u)+x(v)-x(u)}&\ge t, &\text{for } v \in V_W,\\
r(v,u)+x(v)-x(u)&\ge t, &\text{for $u$ s.t. $(v, u) \in E$ for }v \in V_B,  \label{e1-2}\\
\sum_{u \mid (v, u) \in E}   p(v,u)(r(v,u)+x(v)-x(u))&\ge t, &\text{for } v \in V_R,\label{e1-3}
\end{align}
where the constant $b$ will be determined later.
Defining the new variables $y(v):=b^{-x(v)}$, we can rewrite \raf{e1-1}-\raf{e1-3} as follows:

\begin{align}\label{e1-4}
\sum_{u \mid (v, u) \in E}b^{r(v,u)}y(u)&\ge b^ty(v), &\text{for } v \in V_W,\\
b^{r(v,u)}y(u)&\ge b^ty(v), &\text{for $u$ s.t. $(v, u) \in E$ for }v \in V_B,  \label{e1-5}\\
\prod_{u \mid (v, u) \in E}(b^{r(v,u)}y(u))^{p(v,u)}&\ge b^ty(v), &\text{for } v \in V_R.\label{e1-6}
\end{align}
Note that $y(v)>0$ if and only if $x(v)$ is finite. In fact, since we may assume by Theorem~\ref{t2} that $\|x\|_{\infty}\leq L$, we may add also the inequalities:
\begin{align*} 
b^{-L}\le y(v)\le b^{L}, \qquad\text{ for } v \in V.
\end{align*}
Note that, without the lower bounds $y(v)\ge b^{-L}$, the system \raf{e1-4}-\raf{e1-6} is always feasible. As we shall see later, it will be necessary to test the feasibility of the system with  $y(v)>0$ for  some $v\in V$. For convenience, let us write more generally the following set of upper and lower bounds, where $V'\subseteq V$ is to be chosen later:
\begin{align}\label{e1-7}
	0\le y(v)\le b^{L}, \qquad\text{ for } v \in V, \text{ and }
	y(v) \ge b^{-L} , \qquad\text{ for } v \in V'. 
\end{align}

Similarly, we replace the min operator in the system $M[r_x]\le t$ by the softmin approximation:
\begin{align}\label{e2-1}
r(v,u)+x(v)-x(u)&\le t, &\text{for $u$ s.t. $(v, u) \in E$ for }v \in V_W,  \\
-\log_b\sum_{u \mid (v, u) \in E}b^{-r(v,u)-x(v)+x(u)}&\le t, &\text{for } v \in V_B,\label{e2-2}\\
\sum_{u \mid (v, u) \in E}   p(v,u)(r(v,u)+x(v)-x(u))&\le t, &\text{for } v \in V_R,\label{e2-3}
\end{align}
and defining the new variables $y(v):=b^{x(v)}$, we can rewrite \raf{e2-1}-\raf{e2-3} as follows:

\begin{align}\label{e2-4}
b^{-r(v,u)}y(u)&\ge b^{-t}y(v), &\text{for $u$ s.t. $(v, u) \in E$ for }v \in V_W,  \\
\sum_{u \mid (v, u) \in E}b^{-r(v,u)}y(u)&\ge b^{-t}y(v), &\text{for } v \in V_B,\label{e2-5}\\
\prod_{u \mid (v, u) \in E}(b^{-r(v,u)}y(u))^{p(v,u)}&\ge b^{-t}y(v), &\text{for } v \in V_R,\label{e2-6}
\end{align}
together with the lower and upper bounds:
\begin{align}\label{e2-7}
0\le y(v)\le b^{L}, \qquad\text{ for } v \in V,\text{ and }
y(v) \ge b^{-L} , \qquad\text{ for } v \in V'. 
\end{align}
\section{Solving the convex programs}
We will use the ellipsoid method \cite{K80,K84,GLS88}. For this we need to show that the separation problem can be solved in polynomial time. For convenience, let us consider the following relaxation of the convex programs~\raf{e1-4}-\raf{e1-7} and~\raf{e2-4}-\raf{e2-7}:
\begin{align}\label{e1-4-}
\sum_{u \mid (v, u) \in E}b^{r(v,u)}y(u)&\ge b^{t} y(v)-\delta, &\text{for } v \in V_W,\\
b^{r(v,u)}y(u)&\ge b^{t} y(v)-\delta, &\text{for $u$ s.t. $(v, u) \in E$ for }v \in V_B,  \label{e1-5-}\\
\prod_{u \mid (v, u) \in E}(b^{r(v,u)}y(u))^{p(v,u)}&\ge b^{t}y(v)-\delta, &\text{for } v \in V_R.\label{e1-6-}
\end{align}
\begin{align}\label{e1-7-}
0\le y(v)\le b^{L}+\delta, \qquad\text{ for } v \in V,\text{ and }
b^{L}y(v)\ge 1, \qquad\text{ for } v \in V'. 
\end{align}

\begin{align}\label{e2-4-}
\sum_{u \mid (v, u) \in E}b^{-r(v,u)}y(u)&\ge b^{-t} y(v)-\delta, &\text{for } v \in V_W,\\
b^{-r(v,u)}y(u)&\ge b^{-t} y(v)-\delta, &\text{for $u$ s.t. $(v, u) \in E$ for }v \in V_B,  \label{e2-5-}\\
\prod_{u \mid (v, u) \in E}(b^{-r(v,u)}y(u))^{p(v,u)}&\ge b^{-t} y(v)-\delta, &\text{for } v \in V_R.\label{e2-6-}
\end{align}
\begin{align}\label{e2-7-}
0\le y(v)\le b^{L}+\delta, \qquad\text{ for } v \in V,\text{ and }
b^{L}y(v) \ge 1, \qquad\text{ for } v \in V'. 
\end{align}
where $\delta>0$ is a rational number that will be chosen appropriately. Let $K$ and $K_\delta$ be the set of $y\in\RR^{E}$ satisfying \raf{e1-4}-\raf{e1-7} and
\raf{e1-4-}-\raf{e1-7-}, respectively. Similarly, Let $K'$ and $K'_\delta$ be the set of $y\in\RR^{E}$ satisfying \raf{e2-4}-\raf{e2-7} and
\raf{e2-4-}-\raf{e2-7-}, respectively. 

In our application, we will set $b=n^{4\Lambda}$ and $t\in[0,U]$ to be a rational number with denominator $\Lambda:=\sqrt{k}2^{k/2}D^{k+1}$. In particular, $b^{\pm t}$ is a rational number of bit length $\langle b^{\pm t}\rangle=O(\Lambda \log n)$. Also by assuming without loss of generality (by scaling $r$ and replacing $U$ by $UD$) that $r(v,u)$ is a multiple of $D$, $b^{\pm r(v,u)}$ is 
a rational number of bit length $\langle b^{\pm r(v,u)}\rangle=O(\Lambda UD\log n)$.
\begin{claim}\label{cl1}
For $0<\epsilon\le b^{-t}\delta$ and any $y\in K$ (resp., $y\in K'$), the box $\{y'\in\RR^n\mid~y\le y'\le y+\epsilon\cdot\bone\}$ is contained in $K_\delta$ (resp., $K'_{\delta}$), where $\bone$ is the $n$-dimensional vector of all ones. In particular, if $K\neq\emptyset$ (resp., $K'\neq\emptyset$) then $K_\delta$ (resp., $K'_\delta$) is full-dimensional. 
\end{claim}
\proof
We prove the statement for $K_\delta$; the proof for $K'_\delta$ is similar. 
Clearly, $K\subseteq K_\delta$ for $\delta\ge 0$. Furthermore, for any $y\in K$ and $S\subseteq V$, the vector $y'$ obtained from $y$ by setting $y'(u):=y(u)+\epsilon$ for $u\in S$ and $y'(u):=y(u)$ for $u\in V\setminus S$ satisfies \raf{e1-4}-\raf{e1-7}. Indeed, the left-hand sides of \raf{e1-4-}-\raf{e1-6-} increase when $y$ is increased in the components corresponding to $S$, while the right-hand sides are at most $b^{t}y(v)+b^{t}\epsilon-\delta\le b^{t}y(v)$. Also by \raf{e1-7}, $y'(v)\le y(v)+\epsilon\le b^L+\epsilon\le b^L+\delta$, so $y'$ satisfies \raf{e1-4-}-\raf{e1-7-}. 
\qed
	
Now we consider the (semi-weak) separation problem for $K_\delta$ (resp., $K_\delta'$):

\medskip

Given $\bar y\in \QQ^n$ and $0<\delta'\in \QQ$, either assert that $\bar y\in K_\delta$ (resp., $\bar y\in K_\delta'$) or find a vector $c\in\QQ^n$ such that $c^Ty+\delta'\ge c^T\bar y$ for all $y\in K_\delta$ (resp., $y\in K_\delta'$).

\medskip

\begin{claim}\label{cl2}
The separation problems for $K_\delta$ and $K_\delta'$ can be solved in $\poly(U,D,\Lambda,\log n,\langle\bar y\rangle,\langle\delta'\rangle)$ time. 	
\end{claim} 	
\proof
We present the proof for $K_\delta$; the proof for $K_\delta'$ is similar.
Clearly, we can check in $\poly(\langle \bar y\rangle,\langle b^t \rangle,\langle \delta\rangle,n)$ time if $\bar y$ satisfies the linear inequalities \raf{e1-4-}, \raf{e1-5-} and \raf{e1-7-}; if one is violated, the corresponding hyperplane defines a (exact) separator $c\in\QQ^n$ and we are done. Assume therefore that $\bar y$ satisfies \raf{e1-4-}, \raf{e1-5-} and \raf{e1-7-}.
Let us now consider an inequality of the form \raf{e1-6-} corresponding to $v\in V$ violated by $\bar y$. Let $f(y)=\prod_{u \mid (v, u) \in E}(b^{r(v,u)}y(u))^{p(v,u)}-b^{t}y(v):=A\cdot g(y)-B\cdot y(v)$, where $A:=\prod_{u \mid (v, u) \in E}b^{r(v,u)p(v,u)}$, $B:=b^t$ and $g(y):=\prod_{u \mid (v, u) \in E}y(u)^{p(v,u)}$. Note that $f(\bar y)< -\delta$ if and only if $A^D\cdot g(\bar y)^D <(B\cdot \bar y(v)-\delta)^D$, which can be checked in $\poly(\langle \bar y\rangle,\langle \delta\rangle,n,\log D)$ time, as both $A^D$ and $B^D$ are non-negative integers\footnote{in case of $K_\delta'$, they are rational numbers of denominator at most $n^{\Lambda D^2U}$}.
Without loss of generality\footnote{For this part of the proof we can consider the restriction of $y$ to the set of positions reachable from  $v$ by one move. We can also replace parallel edges by one edge; if $v$ is a position of chance then the transition probability of this edge is the sum of the transition probabilities of all corresponding parallel edges.}, we assume that for every $u\in V$ there is exactly one edge $(v,u)\in E$. Then $\nabla f(y):=A\cdot\left(\frac{p(v,u)}{y(u)}:~u \in V\right) g(y)- B \bone_v$, where $\bone_v$ is the unit dimensional vector with $1$ in position $v$. Then the inequality
$
-\delta\le f(y)\le f(\bar y)+\nabla f(\bar y)^T(y-\bar y)<-\delta+\nabla f(\bar y)^T(y-\bar y),
$
valid for all $y\in K_\delta$ by concavity of $f(y)$, 
gives a separating inequality: 
\begin{equation}\label{e16-}
\nabla f(\bar y)^Ty>\nabla f(\bar y)^T\bar y~~~\text{ for all $y\in K_\delta$}.
\end{equation}

Note that the vector $\nabla f(\bar y)$ can be irrational (it is irrational whenever $g(\bar y)$ is). We define a rational approximation $\tilde g$ such that $\tilde g \ge g(\bar y)\ge \tilde g-\frac{\delta'}{A}$ and $ c:=A\cdot\left(\frac{p(v,u)}{\bar y(u)}:~u \in V\right) \tilde g- B \bone_v$. 
Since $r(v,u)$ is assumed to be integer multiple of $D$, $A$ is an integer and hence $\tilde g$ is a rational number of bit length $\langle \tilde g\rangle=\langle A\rangle+\langle\delta'\rangle$. It follows also that $c$ is a rational vector of bit length $\poly(UD\Lambda\log n,\langle\bar y\rangle,\langle\delta'\rangle)$.
Note that 
\begin{equation}\label{e12-}
c^Ty-\nabla f(\bar y)^Ty=A\cdot\left(\frac{p(v,u)}{\bar y(u)}:~u \in V\right)^Ty\cdot(\tilde g- g(\bar y))\ge 0 \text{ for all }y\in K_{\delta},
\end{equation}
while
\begin{align}\label{e14-}
c^T\bar y-\nabla f(\bar y)^T\bar y&=A\cdot\left(\frac{p(v,u)}{\bar y(u)}:~u \in V\right)^T\bar y\cdot(\tilde g- g(\bar y))\le \delta'
\end{align}
It follows from \raf{e16-}, \raf{e12-}, and \raf{e14-} that $c^Ty+\delta'\ge c^T\bar y$ for all $y\in K_{\delta}$.   
\qed
 
\begin{lemma}\label{l2-}
	Given $t\in \RR$ and $\delta\in (0,1)$ we can decide in time $\poly(n,U,D^k,\log\frac{1}{\delta})$ if the system \raf{e1-4}-\raf{e1-7} (resp., \raf{e2-4}-\raf{e2-7}) is infeasible, or find $y(v)\in[b^{-L},b^L+\delta]$, for $v\in V'$ and $y(v)\in[0,b^L+\delta]$, for $v\in V\setminus V'$, such that the left hand sides of \raf{e1-4}-\raf{e1-6} (resp., \raf{e2-4}-\raf{e2-6}) are at least $b^ty(v)-\delta$ (resp., $b^{-t}y(v)-\delta$), for all $v\in V$.
\end{lemma}
\begin{proof} 
Let us consider the equivalent system~\raf{e1-4}-\raf{e1-7}; the proof for~\raf{e2-4-}-\raf{e2-7-} is similar. 
Given a polynomial-time algorithm for the separation problem for the convex set $K_\delta$, a circumscribing ball of radius $H$ for $K_\delta$, and any $\epsilon'>0$, the ellipsoid method terminates in $N:=O(n\log\frac{1}{\epsilon'}+n^2|\log H|)$ calls to the separation algorithm using $\delta'=2^{-O(N)}$, and either (i) finds a vector $y\in K_\delta$, or (ii) asserts that $\vol(K_\delta)\le\epsilon'$;  
see, e.g., Theorem~3.2.1 in \cite{GLS88}. In the first case, we get a vector $y$ satisfying the conditions in the statement of the lemma. In the second case, we conclude that $K_\delta$ and hence $K$ is empty if $\epsilon'<(b^{-t}\delta)^n$. Indeed by Claim~\ref{cl1}, if $K\neq\emptyset$ and $\epsilon:=b^{-t}\delta$, then $\vol(K_\delta)\ge\epsilon^n>\epsilon'$, given a contradiction to the assertion in (ii). 

By \raf{e1-7-}, the radius of the bounding ball can be chosen as $H:=2b^L$.  
Furthermore, the ellipsoid method works only with numbers having precision of $O(N)$ bits. By Claim~\ref{cl2}, the separation problem can be solved in time $\poly(n,U,D^k,\log\frac{1}{\delta})$. 
\qed	
\end{proof}

\begin{remark}\label{r1}
By raising to inequalities~\raf{e1-6} and \raf{e2-6} to power $D$, we obtain systems of polynomial inequalities. Khachiyan \cite{K83,K84} gave a polynomial-time algorithm for (approximately) solving a system of convex polynomial inequalities. However, it is not possible to use this algorithm directly to solve the convex programs~\raf{e1-4}-\raf{e1-7} and~\raf{e2-4}-\raf{e2-7}, since the polynomials obtained after raising inequalities~\raf{e1-6} and \raf{e2-6} to power $D$ are not necessarily convex. For instance, the function $\sqrt{xy}-z$ is concave for $x,y,z\in\RR_+$, while the function $xy-z^2$ is not.	
\end{remark}

\section{A Pseudo-polynomial algorithm for $k=O(1)$} \label{sec111}
Let $\cG$ be a BWR-game. Let $t_{\max}:=\max_{v\in V}\mu^{}_{\cG}(v)$ and $t_{\min}:=\min_{v\in V}\mu^{}_{\cG}(v).$ Define the {\it top} and {\it bottom} classes of $\cG$ as $\cT:=\{v\in V~|~\mu^{}_{\cG}(v)=t_{\max}\}$ and $\cB:=\{v\in V~|~\mu^{}_{\cG}(v)=t_{\min}\}$, respectively. 

\begin{proposition}\label{p1}
	Top and bottom classes necessarily satisfy the following properties.	 
	\begin{itemize}
		\item[(i)] There exists no arc $(v,u)\in E$ such that $v\in (V_W\cup V_R)\cap \cB$, $u\not\in \cB$;
		\item[(ii)] there exists no arc $(v,u)\in E$ such that $v\in( V_B\cup V_R)\cap \cT$, $u\not\in \cT$;
		\item[(iii)] there exists no arc $(v,u)\in E$ such that $v\in V_W\setminus\cT$, $u\in \cT$;
		\item[(iv)] there exists no arc $(v,u)\in E$ such that $v\in V_B\setminus\cB$, $u\in \cB$;		
		\item[(v)] for every $v\in V_W\cap \cT$, there exists an arc $(v,u)\in E$ such that $u\in \cT$;
		\item[(vi)] for every $v\in V_B\cap \cB$, there exists an arc $(v,u)\in E$ such that $u\in \cB$;
		\item[(vii)] for every $v\in (V_B\cup V_R)\setminus \cT$, there exists an arc $(v,u)\in E$ such that $u\not\in \cT$;
		\item[(viii)] for every $v\in (V_W\cup V_R)\setminus \cB$, there exists an arc $(v,u)\in E$ such that $u\not\in \cB$.
	\end{itemize}
\end{proposition}
\proof
All claims follow from the existence of a canonical form for $\cG$, by Theorem \ref{t2}. Indeed, the existence of arcs forbidden by (i), (ii), (iii) and (iv), or the non-existence of arcs required by (v), (vi), (vii) and (viii) would violate the value equations (C1) of the canonical form.
\qed

\begin{lemma}\label{l3-1}
Consider the convex program defined by~\raf{e1-4}-\raf{e1-6} (resp., ~\raf{e2-4}-\raf{e2-6}). Then for $t:=t_{\max}$ (resp., $t:=t_{\min}$), there is a feasible solution with $y(v)\ge b^{-L}$ for all $v\in\cT$ (resp., $v\in\cB$). 	
\end{lemma}
\proof
Consider the game $\cG[\cT]$ (resp., $\cG[\cB]$) induced by the top class $\cT$ (resp., the bottom class $\cB$).  
Let $x\in\RR^\cT$ (resp., $x\in\RR^\cB$) be the potential vector guaranteed by Theorem~\ref{t2} for the game $\cG[\cT]$ (resp., $\cG[\cB]$). Set $y(v):=b^{-x(v)}$ for $v\in\cT$ (resp., $y(v):=b^{x(v)}$ for $v\in\cB$) and $y(v)=0$ for $v\in V\setminus\cT$ (resp., $v\in V\setminus\cB$). Then $y(v)\ge b^{-L}$ for all $v\in\cT$ (resp., $v\in\cB$). It is easy to verify by Proposition~\ref{p1} that the system is feasible. Indeed, \raf{e1-4} is satisfied for every position $v\in V_W\cap\cT$ (resp., $v\in V_B\cap\cB$) by the definition of $x$ and Fact~\ref{softmax-softmin}: 
$$t_{\max}\le \max_{u\mid (v,u)\in E}(r(v,u)+x(v)-x(u))\le\log_b\sum_{u \mid (v, u) \in E}b^{r(v,u)+x(v)-x(u)}$$ 

$$\text{(resp., }-t_{\min}\le -\min_{u\mid (v,u)\in E}(r(v,u)+x(v)-x(u))\le-\log_b\sum_{u \mid (v, u) \in E}b^{-r(v,u)-x(v)+x(u)}\text{ ).}$$ 
Moreover, for $v\in (V_B\cup V_R)\cap \cT$ (resp., $v\in (V_W\cup V_R)\cap \cB$) we have \raf{e1-5} and \raf{e1-6} (resp., \raf{e2-5} and \raf{e2-6}) satisfied by the definition of $x$ and Proposition~\ref{p1}-(ii) (resp., Proposition~\ref{p1}-(i)), while for $v\in V\setminus\cT$ (resp., $v\in V\setminus\cB$) \raf{e1-4}-\raf{e1-6} (resp., \raf{e2-4}-\raf{e2-6}) are trivially satisfied.  
\qed 

\medskip

In the following we set $\delta(t):=\frac{1}{2}b^{-t-L}(1-\frac{1}{n})$, where $\eps:=\frac{1}{\sqrt{k}2^{k/2+1}D^{k+1}}$. Note that $b:=n^{\frac{2}{\eps}}$.
\begin{lemma}\label{l33-2} 
	The values $t_{\max}$ and $t_{\min}$ can be found in time $\poly(n,U,D^k)$.   
\end{lemma}
\proof
    We only show how to find $t_{\max}$; in a similar fashion we can determine 
    $t_{\min}$. 
	We apply Lemma~\ref{l2-} in a binary search manner to check the feasibility of the system~\raf{e2-4}-\raf{e2-7} for $t\in[0,U]$ and $\delta(t)$ as specified above. 
	Note that, by Lemma~\ref{l1-prob}, $t_{max}\in[0,U]$ can be written as a rational number with denominator at most $\frac{1}{2\eps}$. So we may restrict our search steps to integer multiples of $\frac{1}{2\eps}$. We stop the search when the length of the search interval becomes a constant multiple of $\frac{1}{2\eps}$, and then apply linear search for the remaining small interval.   
	
	Suppose that the convex program~\raf{e2-4}-\raf{e2-7} is infeasible. Then Theorem~\ref{t2} implies that $t_{\max}>t$. On the other hand,	if $y\in\RR^V$ is a $\delta(t)$-approximately feasible solution for~\raf{e2-4}-\raf{e2-7}, then as $\delta(t)\le \frac{1}{2}b^{-L}$, the new vector $y':=2y$ satisfies $y'(v)\in[b^{-L},2b^L+b^{-L}]$ for all $v\in V$. Also, $y'$ satisfies~\raf{e2-4}-\raf{e2-6} within an error of $2\delta(t)$, that is, the left-hand sides of \raf{e2-4}-\raf{e2-6}, when $y$ is replaced by $y'$, are at least $b^{-t}y'(v)-2\delta(t)=b^{-t}y'(v)-b^{-t-L}(1-\frac{1}{n})\ge b^{-t-\log_bn}y'(v)$. Set $x(v):=\log_b y'(v).$ Then $x$ satisfies \raf{e2-1}-\raf{e2-3} with $t$ replaced by $t+\log_bn$. This in turn implies by Fact~\ref{softmax-softmin} that $M[r_{x}]\le (t+2\log_b n)\bone=(t+\eps)\bone$. 
	It follows then from Lemma~\ref{l2} that $t_{\max}\le t+\eps$.
	Recall that we assume both $t$ and $t_{max}$ are multiples of ${2\eps}$; hence, $t_{\max}\le t$.
	
	Since the number of binary search steps is at most $\log \frac{U}{2\epsilon}=O(k\log (UD))$ and each step requires time $\poly(n,U,\log b,\log\frac{1}{\delta})=\poly(n,U,^k)$, the bound on the running time follows.
\qed

\begin{lemma}
	We can find the top class $\cT$ (resp., bottom class $\cB$) in time $\poly(n,U,D^k)$.   
\end{lemma}
\proof
We can check if a vertex $w\in V$ belongs to the top class (resp., bottom class) as follows. We write the convex program~\raf{e1-4}-\raf{e1-6} (resp.,~\raf{e2-4}-\raf{e2-6}) with $t:=t_{max}$ (resp., $t=t_{\min}$) and with the additional constraint that $y(w)\ge b^{-L}$ and $y(v)\ge 0$ for all $v\in V\setminus\{w\}$. Then we check the feasibility of this system.
If the system is infeasible then we know by Lemma~\ref{l3-1} that $w\not\in\cT$ (resp., $w\not\in\cB$).

Suppose, on the other hand, that $y\in\RR^V$ is a $\delta(t_{\max})$-approximately (resp., $\delta(t_{\min})$-approximately) feasible solution for~\raf{e1-4}-\raf{e1-6}  (resp.,~\raf{e2-4}-\raf{e2-6}). Then as in the proof of Lemma~\ref{l33-2}, the new vector $y':=2y$ satisfies $y'(w)\ge b^{-L}>0$, and the left-hand sides of \raf{e1-4}-\raf{e1-6} (resp., \raf{e2-4}-\raf{e2-6}), when $y$ is replaced by $y'$, are at least $b^{t-\log_bn}y'(w)$ (resp., $b^{-t-\log_bn}y'(w)$).

Now we claim that $V^+:=\{v\in V:~y(v)>0\}\subseteq\cT$ (resp., $V^+:=\{v\in V:~y(v)>0\}\subseteq\cB$), which would in turn imply that $w\in\cT$ (resp., $w\in\cB$). Indeed, constraints \raf{e1-4}-\raf{e1-6} (resp.,~\raf{e2-4}-\raf{e2-6}), applied to $y$ replaced by $y'$, imply that (i) if $v\in V_W\cap V^+$ then there exists an arc $(v,u)\in E$ such that $u\in V^+$; (ii) if $v\in (V_B\cup V_R)\cap V^+$ then all arcs $(v,u)\in E$ must that $u\in V^+$ (resp., (i) if $v\in V_B\cap V^+$ then there exists an arc $(v,u)\in E$ such that $u\in V^+$; (ii) if $v\in (V_W\cup V_R)\cap V^+$ then all arcs $(v,u)\in E$ must that $u\in V^+$).
These imply that the game induced by $V^+$ is well-defined and, by Lemma~\ref{l1} (resp., Lemma~\ref{l2}), all its positions have value at least $t_{\max}$ (resp., at most $t_{min}$). 
The lemma follows. 
\qed

Finally, given the top and bottom classes, we can find an optimal pair of strategies in the games induced by $\cT$ and $\cB$, as stated in the next lemma.

\begin{lemma}
We can find optimal pairs of strategies in the games induced by the top class $\cT$ and bottom class $\cB$ in time $\poly(n,U,D^k)$.   
\end{lemma} 
\proof
We prove the lemma only for $\cT$; the proof for $\cB$ can be done similarly. 
We solve two (feasible) systems, $\cS_1$ defined by~\raf{e1-4}-\raf{e1-7} on $\cG[\cT]$ and $\cS_2$ defined by~\raf{e2-4}-\raf{e2-6} on $\cG[\cT]$, with $t:=t_{\max}$ to within an accuracy of $\delta(t_{\max})$. Let $y^1,y^2\in\RR^\cT$  be the $\delta(t_{\max})$-approximate solutions to $\cS_1$ and $\cS_2$, respectively. By the same arguments as in Lemma~\ref{l33-2}, the corresponding potential vectors $x^1,x^2$ (defined by $x^1(v):=-\log_b (2y^1(v))$ and $x^2(v):=\log_b (2y^2(v))$) ensure that $M[r_{x^1}]\ge (t_{\max}-\eps)\bone$ and $M[r_{x^2}]\le (t_{\max}+\eps)\bone$. Since $\eps$ is sufficiently small, by Lemmas~\ref{l1} and \ref{l2}, the locally optimal strategies defined by the operator $M$ with respect to $x^1$ and $x^2$ give optimal strategies for $\Max$ and $\Min$ in $\cG[\cT]$, respectively.
 \qed

Finally, we obtain Theorem~\ref{t-main} by combining the above lemmas with the algorithm in \cite{BEGM13,BEGM15-arxiv}.




\begin{thebibliography}{BEGM13b}
	
	\bibitem[AM09]{AM09}
	D.~Andersson and P.~B. Miltersen.
	\newblock The complexity of solving stochastic games on graphs.
	\newblock In {\em ISAAC}, pages 112--121, 2009.
	
	\bibitem[BEGM10]{IPCO-2010}
	E.~Boros, K.~Elbassioni, V.~Gurvich, and K.~Makino.
	\newblock A pumping algorithm for ergodic stochastic mean payoff games with
	perfect information.
	\newblock In {\em Proc. 14th IPCO}, pages 341--354, 2010.
	
	\bibitem[BEGM13a]{DGAA13}
	E.~Boros, K.~Elbassioni, V.~Gurvich, and K.~Makino.
	\newblock On canonical forms for zero-sum stochastic mean payoff games.
	\newblock {\em Dynamic Games and Applications}, 3(2):128--161, 2013.
	
	\bibitem[BEGM13b]{BEGM13}
	Endre Boros, Khaled~M. Elbassioni, Vladimir Gurvich, and Kazuhisa Makino.
	\newblock A pseudo-polynomial algorithm for mean payoff stochastic games with
	perfect information and a few random positions.
	\newblock In {\em Automata, Languages, and Programming - 40th International
		Colloquium, {ICALP} 2013, Riga, Latvia, July 8-12, 2013, Proceedings, Part
		{I}}, pages 220--231, 2013.
	
	\bibitem[BEGM15]{BEGM15-arxiv}
	E.~Boros, K.~Elbassioni, V.~Gurvich, and K.~Makino.
	\newblock A pseudo-polynomial algorithm for mean payoff stochastic games with
	perfect information and few random positions.
	\newblock {\em CoRR}, abs/1508.03431, 2015.
	
	\bibitem[BV01a]{BV01}
	E.~Beffara and S.~Vorobyov.
	\newblock Adapting gurvich-karzanov-khachiyan's algorithm for parity games:
	Implementation and experimentation.
	\newblock {Technical Report 2001-020}, Department of Information Technology,
	Uppsala University, available at:
	https://www.it.uu.se/research/reports/$\#$2001, 2001.
	
	\bibitem[BV01b]{BV01a}
	E.~Beffara and S.~Vorobyov.
	\newblock Is randomized gurvich-karzanov-khachiyan's algorithm for parity games
	polynomial?
	\newblock {Technical Report 2001-025}, Department of Information Technology,
	Uppsala University, available at:
	https://www.it.uu.se/research/reports/$\#$2001, 2001.
	
	\bibitem[BV05]{BV05}
	H.~Bj\"{o}rklund and S.~Vorobyov.
	\newblock Combinatorial structure and randomized sub-exponential algorithm for
	infinite games.
	\newblock {\em Theoretical Computer Science}, 349(3):347--360, 2005.
	
	\bibitem[BV07]{BV07}
	H.~Bj\"{o}rklund and S.~Vorobyov.
	\newblock A combinatorial strongly sub-exponential strategy improvement
	algorithm for mean payoff games.
	\newblock {\em Discrete Applied Mathematics}, 155(2):210--229, 2007.
	
	\bibitem[CH08]{CH08}
	K.~Chatterjee and T.~A. Henzinger.
	\newblock Reduction of stochastic parity to stochastic mean-payoff games.
	\newblock {\em Inf. Process. Lett.}, 106(1):1--7, 2008.
	
	\bibitem[CHKN14]{Vienna}
	K.~Chatterjee, M.~Henzinger, S.~Krinninger, and D.~Nanongkai.
	\newblock Polynomial-time algorithms for energy games with special weight
	structures.
	\newblock {\em Algorithmica}, 70(3):457--492, 2014.
	
	\bibitem[CJH04]{CJH04}
	K.~Chatterjee, M.~Jurdzi\'{n}ski, and T.~A. Henzinger.
	\newblock Quantitative stochastic parity games.
	\newblock In {\em SODA '04}, pages 121--130, Philadelphia, PA, USA, 2004.
	Society for Industrial and Applied Mathematics.
	
	\bibitem[EM79]{EM79}
	A.~Eherenfeucht and J.~Mycielski.
	\newblock Positional strategies for mean payoff games.
	\newblock {\em International Journal of Game Theory}, 8:109--113, 1979.
	
	\bibitem[Gil57]{Gil57}
	D.~Gillette.
	\newblock Stochastic games with zero stop probabilities.
	\newblock In A.W.~Tucker M.~Dresher and P.~Wolfe, editors, {\em Contribution to
		the Theory of Games III, in Annals of Mathematics Studies}, volume~39, pages
	179--187. Princeton University Press, 1957.
	
	\bibitem[GKK88]{GKK88}
	V.~Gurvich, A.~Karzanov, and L.~Khachiyan.
	\newblock Cyclic games and an algorithm to find minimax cycle means in directed
	graphs.
	\newblock {\em USSR Computational Mathematics and Mathematical Physics},
	28:85--91, 1988.
	
	\bibitem[GLS88]{GLS88}
	M.~Gr\"{o}tschel, L.~Lov\'{a}sz, and A.~Schrijver.
	\newblock {\em Geometric Algorithms and Combinatorial Optimization}.
	\newblock Springer, New York, 1988.
	
	\bibitem[Hal07]{Hal07}
	N.~Halman.
	\newblock Simple stochastic games, parity games, mean payoff games and
	discounted payoff games are all \text{LP}-type problems.
	\newblock {\em Algorithmica}, 49(1):37--50, 2007.
	
	\bibitem[HBV04]{BSV04}
	S.~Sandberg H.~Bj\"{o}rklund and S.~Vorobyov.
	\newblock A combinatorial strongly sub-exponential strategy improvement
	algorithm for mean payoff games.
	\newblock {DIMACS Technical Report 2004-05}, DIMACS, Rutgers University, 2004.
	
	\bibitem[JPZ06]{JPZ06}
	M.~Jurdzi\'{n}ski, M.~Paterson, and U.~Zwick.
	\newblock A deterministic subexponential algorithm for solving parity games.
	\newblock In {\em SODA '06}, pages 117--123, New York, NY, USA, 2006. ACM.
	
	\bibitem[Jur98]{J98}
	M.~Jurdzi\'{n}ski.
	\newblock Deciding the winner in parity games is in \text{UP} $\cap$
	\text{co-UP}.
	\newblock {\em Inf. Process. Lett.}, 68(3):119--124, 1998.
	
	\bibitem[Jur00]{J2000}
	M.~Jurdzi\'{n}ski.
	\newblock {\em Games for Verification: Algorithmic Issues}.
	\newblock PhD thesis, Department of Computer Science, University of Aarhus,
	Denmark, 2000.
	
	\bibitem[Kar78]{Kar78}
	R.~M. Karp.
	\newblock A characterization of the minimum cycle mean in a digraph.
	\newblock {\em Discrete Math.}, 23:309--311, 1978.
	
	\bibitem[Kha80]{K80}
	L.~Khachiyan.
	\newblock Polynomial algorithms in linear programming.
	\newblock {\em U.S.S.R. Comput. Math. Math. Phys.}, 20:53--72, 1980.
	
	\bibitem[Kha83]{K83}
	L.~Khachiyan.
	\newblock Convexity and complexity in polynomial programming.
	\newblock In {\em Proceedings of the International Congress of Mathematicians,
		Warsaw}, page 1569â€“1577, 1983.
	
	\bibitem[Kha84]{K84}
	L.~Khachiyan.
	\newblock {\em Complexity of convex problems of real and integer programming}.
	\newblock PhD thesis, The second thesis, manuscript, Computer Center of the
	USSR Academy of Sciences (in Russian), Moscow, 1984.
	
	\bibitem[KL93]{KL93}
	A.~V. Karzanov and V.~N. Lebedev.
	\newblock Cyclical games with prohibition.
	\newblock {\em Mathematical Programming}, 60:277--293, 1993.
	
	\bibitem[Lit96]{Lit96}
	M.~L. Littman.
	\newblock {\em Algorithm for sequential decision making, CS-96-09}.
	\newblock PhD thesis, Dept. of Computer Science, Brown Univ., USA, 1996.
	
	\bibitem[LL69]{LL69}
	T.~M. Liggett and S.~A. Lippman.
	\newblock Stochastic games with perfect information and time-average payoff.
	\newblock {\em SIAM Review}, 4:604--607, 1969.
	
	\bibitem[MO70]{MO70}
	H.~Mine and S.~Osaki.
	\newblock {\em Markovian decision process}.
	\newblock American Elsevier Publishing Co., New York, 1970.
	
	\bibitem[Mou76a]{Mou76a}
	H.~Moulin.
	\newblock Extension of two person zero sum games.
	\newblock {\em Journal of Mathematical Analysis and Application},
	5(2):490--507, 1976.
	
	\bibitem[Mou76b]{Mou76}
	H.~Moulin.
	\newblock Prolongement des jeux \`{a} deux joueurs de somme nulle.
	\newblock {\em Bull. Soc. Math. France, Memoire}, 45, 1976.
	
	\bibitem[Pis99]{Pis99}
	N.~N. Pisaruk.
	\newblock Mean cost cyclical games.
	\newblock {\em Mathematics of Operations Research}, 24(4):817--828, 1999.
	
	\bibitem[RF91]{RF91}
	T.~E.~S. Raghavan and J.~A. Filar.
	\newblock Algorithms for stochastic games: A survey.
	\newblock {\em Mathematical Methods of Operations Research}, 35(6):437--472,
	1991.
	
	\bibitem[Sch09]{Schewe09}
	Sven Schewe.
	\newblock From parity and payoff games to linear programming.
	\newblock In {\em Mathematical Foundations of Computer Science 2009, 34th
		International Symposium, {MFCS} 2009, Novy Smokovec, High Tatras, Slovakia,
		August 24-28, 2009. Proceedings}, pages 675--686, 2009.
	
	\bibitem[Vor08]{Vor08}
	S.~Vorobyov.
	\newblock Cyclic games and linear programming.
	\newblock {\em Discrete Applied Mathematics, Special volume in Memory of Leonid
		Khachiyan (1952 - 2005)}, 156(11):2195--2231, 2008.
	
	\bibitem[ZP96]{ZP96}
	U.~Zwick and M.~Paterson.
	\newblock The complexity of mean payoff games on graphs.
	\newblock {\em Theoretical Computer Science}, 158(1-2):343--359, 1996.
	
\end{thebibliography}
\end{document}